\documentclass{llncs}
\usepackage{tabularx,booktabs,multirow,delarray,array}
\usepackage{graphicx,amssymb,amsmath,amssymb}
\usepackage{latexsym}
\def\calP{\mathcal{P}}
\newcommand{\VD}{\mbox{$V\!D$}}

\def\calA{\mathcal{A}}

\def\subsectionspace{\vspace*{-0.08in}}

\newenvironment{proof}{\noindent {\textbf{Proof:}}\rm}{\hfill $\Box$\rm}
\newtheorem{observation}{Observation}

\begin{document}

\title{Weak Visibility Queries of Line Segments in Simple Polygons}

\author{Danny Z.~Chen\inst{1}\thanks{The research of D.Z.~Chen
was supported in part by NSF under Grants CCF-0916606
and CCF-1217906.} \and Haitao Wang\inst{2}\thanks{Corresponding
author.
 }}

\institute{
  Department of Computer Science and Engineering\\
  University of Notre Dame, Notre Dame, IN 46556, USA\\
  \email{dchen@cse.nd.edu}\\
  \and
  Department of Computer Science\\
  Utah State University, Logan, UT 84322, USA\\
  \email{haitao.wang@usu.edu}
}
\maketitle

\pagestyle{plain}
\pagenumbering{arabic}
\setcounter{page}{1}

\begin{abstract}
Given a simple polygon $P$ in the plane,
we present new algorithms and data structures for
computing the weak visibility polygon from any query
line segment in $P$. We build a data structure in $O(n)$ time and
$O(n)$ space that can compute the visibility polygon
for any query line segment $s$ in
$O(k\log n)$ time, where $k$ is the
size of the visibility polygon of $s$ and $n$ is the number of
vertices of $P$. Alternatively, we build a data
structure in $O(n^3)$ time and $O(n^3)$ space that can
compute the visibility polygon for any query line segment in
$O(k+\log n)$ time.
%Our results, which improve the previous work, are
%based on new observations on the problem.
\end{abstract}

%\newpage

\section{Introduction}
\label{sec:mainintro}

Given a simple polygon $\calP$ of $n$ vertices in the plane,
two points in $\calP$ are {\em visible} to each other if the line segment joining
them lies in $\calP$. For a line segment $s$ in $\calP$, a point $p$ is
{\em weakly visible} (or {\em visible} for short) to $s$ if $s$ has at least one point that is visible to $p$.
The {\em weak visibility polygon} (or {\em visibility polygon} for short) of $s$, denoted by $Vis(s)$, is
the set of all points in $\calP$ that are visible to
$s$. The {\em weak visibility query problem} is to
build a data structure for $\calP$ such that
$Vis(s)$ can be computed efficiently for any
query line segment $s$ in $\calP$.

This problem has been studied before.
Bose {\em et al.} \cite{ref:BoseEf02} built a data structure of
$O(n^3)$ size in $O(n^3\log n)$ time that can compute $Vis(s)$ in
$O(k\log n)$ time for any query, where $k$ is the size of $Vis(s)$.
Throughout this paper, we always let $k$ denote the size of $Vis(s)$ for any
query line segment $s$. Bygi and Ghodsi \cite{ref:BygiWe11} gave an
improved data structure with the same size and preprocessing time as
that in \cite{ref:BoseEf02} but its query time is $O(k+\log n)$.
Aronov {\em et al.} \cite{ref:AronovVi02} proposed a smaller data
structure of $O(n^2)$ size with $O(n^2\log n)$ preprocessing time and
$O(k\log^2n)$ query time. Table \ref{tab:10} gives a summary.
If the problem is to compute $Vis(s)$ for a single segment (not queries),
then there is an $O(n)$ time algorithm \cite{ref:GuibasLi87}.

\begin{table}[t]
{
\begin{center}
\caption{\footnotesize A summary of the data structures.
The value $k$ is the size of $Vis(s)$ for any query segment $s$.
%and $K=O(n^2)$ is the size of the visibility graph of $\calP$.
}
\label{tab:10}
\begin{tabularx}{0.70\textwidth}{llll}
\toprule
%\multicolumn{1}{c}{}
Data Structure \  & Preprocessing Time \   & Size & Query Time\\
\midrule
\cite{ref:BoseEf02}   & $O(n^3\log n)$  & $O(n^3)$ &   $O(k\log n)$ \\
\cite{ref:BygiWe11} & $O(n^3\log n)$ & $O(n^3)$ & $O(k+\log n)$ \\
\cite{ref:AronovVi02} & $O(n^2\log n)$ & $O(n^2)$ \  & $O(k\log^2n)$ \\
Our Result 1 & $O(n)$ & $O(n)$ & $O(k\log n)$\\
Our Result 2 & $O(n^3)$ & $O(n^3)$ & $O(k+\log n)$\\
\bottomrule
\end{tabularx}
\end{center}
}
\vspace*{-0.2in}
\end{table}

%\sectionspace
\subsection{Our Contributions}

In this paper, we present two new data structures whose performances are
also given in Table \ref{tab:10}.
%Let $K$ be the size of the visibility graph of $\calP$
%\cite{ref:HershbergerAn89}, which is proportional to the number of
%pairs of vertices of $\calP$ that are mutually visible.
%Note that $K=\Omega(n)$ and $K=O(n^2)$.
Our first data structure, which is built in $O(n)$ time and
$O(n)$ space, can compute $Vis(s)$ in $O(k\log n)$ time for any query
segment $s$. Comparing with the data structure in \cite{ref:AronovVi02},
our data structure reduces the query time by a logarithmic factor and
uses much less preprocessing time and space.
%In the cases where $K=\Theta(n)$, our data structure has only $O(n)$
%preprocessing time and space.

The preprocessing time and size of our second data structure are
both $O(n^3)$,
%, where $K$ is the size of the visibility graph of $\calP$,
and each query takes $O(k+\log n)$ time.
%Note that $K=O(n^2)$ and $K=\Omega(n)$.
Comparing with the
result in \cite{ref:BygiWe11}, our data structure has less
preprocessing time. In addition, our solution, which is based on the approach in
\cite{ref:BoseEf02}, is much simpler than that in \cite{ref:BygiWe11}.
Further, our techniques explore many geometric observations on the problem that
may be useful elsewhere. For example, we prove a tight
combinatorial bound for the
``zone" in a line segment arrangement contained in a simple polygon,
as follows, which is
interesting in its own right.

Let $S$ be a set of line segments in a simple polygon
$\calP$ such that both endpoints of
each segment of $S$ are on $\partial\calP$ (i.e., the boundary of
$\calP$).
Let $\calA$ be the arrangement formed by the segments in $S$ and
the edges of $\partial\calP$.
For any line segment $s$ in $\calP$ (the endpoints of $s$ need
not be on $\partial\calP$), the {\em zone} of $s$, denoted by $Z(s)$,
is defined to be
the set of all faces of $\calA$ that $s$ intersects.
For each edge of any face in $\calA$, it
either lies on a segment of $S$ or lies on $\partial\calP$. Let
$\Lambda$ be the number of edges of the faces in $Z(s)$ each of which
lies on a segment of $S$.
We want to find a good upper bound for $\Lambda$.
By using the zone theorem for the general line segment
arrangement \cite{ref:EdelsbrunnerAr92}, we can easily
obtain $\Lambda=O(|S|\alpha(|S|))$,
where $\alpha(\cdot)$ is the functional inverse
of Ackermann's function \cite{ref:HartNo86}.
In this paper, we prove a tight
bound $\Lambda=O(m)$, where $m\leq |S|$ is the number of segments in
$S$ each of which contains at least one edge of the faces in $Z(s)$.
An immediate
application of this result is that we obtain an efficient query
algorithm for our second data
structure. Since
combinatorial bounds on arrangements are fundamental, this result
may find other applications as well.

The rest of this paper is organized as follows.
In Section \ref{sec:pre}, we review some geometric structures and
a query algorithmic scheme that will be used by the query
algorithms of both our data structures. We will also give a
``ray-rotating'' data structure in Section \ref{sec:pre}, which is
needed by our first data structure in Section \ref{sec:first}.
In Sections \ref{sec:first} and \ref{sec:second}, we present our first
and second data structures, respectively.
As a by-product of our second data structure,
the combinatorial bound of the zone mentioned above is also given in
Section \ref{sec:second}.
Section \ref{sec:conclusions} concludes the paper.

\section{Preliminaries}
\label{sec:pre}

In this section, we review some geometric structures and
discuss an algorithmic scheme that will be used by the query
algorithms of both our data structures given in Sections
\ref{sec:first} and \ref{sec:second}.
We will also give a
``ray-rotating'' data structure in Section \ref{sec:rayrotate}, which is
needed by our first data structure in Section \ref{sec:first}.

For simplicity of discussion, we assume no three vertices of
$\calP$ are collinear; we also assume for any query segment $s$,
$s$ is not collinear with
any vertex of $\calP$ and each endpoint of $s$ is not collinear with
any two vertices of $\calP$.
As in \cite{ref:AronovVi02,ref:BoseEf02},
our approaches can be easily extended to the general situation.

Denote by $\partial\calP$ the boundary of $\calP$.
%For any set $R$ of
%points, we say $R$ is in $\calP$ if each point in $R$ is either in the
%interior of $\calP$ or on $\partial\calP$.
The visibility graph of $\calP$ is a graph whose vertex set consists
of all vertices of $\calP$ and whose edge set consists of edges
defined by all visible pairs of vertices of $\calP$. Here, two
adjacent vertices on $\partial\calP$ are considered visible
to each other. In this paper, we always use $K$ to denote the size of
the visibility graph of $\calP$. Note that $K=O(n^2)$ and $K=\Omega(n)$. The visibility
graph can be computed in $O(K)$ time \cite{ref:HershbergerAn89}.

We introduce the {\em visibility decomposition} of $\calP$
\cite{ref:AronovVi02,ref:BoseEf02}. Consider a point $p$ in $\calP$
and a vertex $v$ of $\calP$. Suppose the line segment $\overline{pv}$
is in $\calP$, i.e., $p$ is visible to $v$. We extend $\overline{pv}$
along the direction from $p$ to $v$ and suppose we stay inside
$\calP$ (when this happens, $v$ must be a reflex vertex). Let $w$ be
the point on the boundary of $\calP$ that is hit first by our above
extension of $\overline{pv}$ (e.g., see Fig.~\ref{fig:window}). We
call the line segment $\overline{vw}$ the {\em window} of $p$. The
point $p$ is called the {\em defining point} of the window and the
vertex $v$ is called the {\em anchor vertex} of the window. It is
well known that the boundary of the visibility polygon of the point
$p$ consists of parts of $\partial\calP$ and the windows of $p$
\cite{ref:AronovVi02,ref:BoseEf02}. If the point $p$ is a vertex of $\calP$, then the window $\overline{vw}$ is called a {\em critical constraint} of $\calP$ and $p$ is called the {\em defining vertex} of the critical constraint.
For example, in Fig.~\ref{fig:subproblem}, the two critical
constraints $\overline{up_u}$ and $\overline{vp_v}$ are both defined
by $u$ and $v$; for $\overline{up_u}$, its anchor vertex is $u$ and
its
defining vertex is $v$, and for $\overline{vp_v}$, its anchor vertex is $v$ and
defining vertex is $u$. It is easy
to see that the total number of critical constraints is $O(K)$
because each critical constraint corresponds to a visible vertex pair of
$\calP$ and a visible vertex pair corresponds to at most
two critical constraints.

\begin{figure}[t]
\begin{minipage}[t]{0.49\linewidth}
\begin{center}
\includegraphics[totalheight=0.8in]{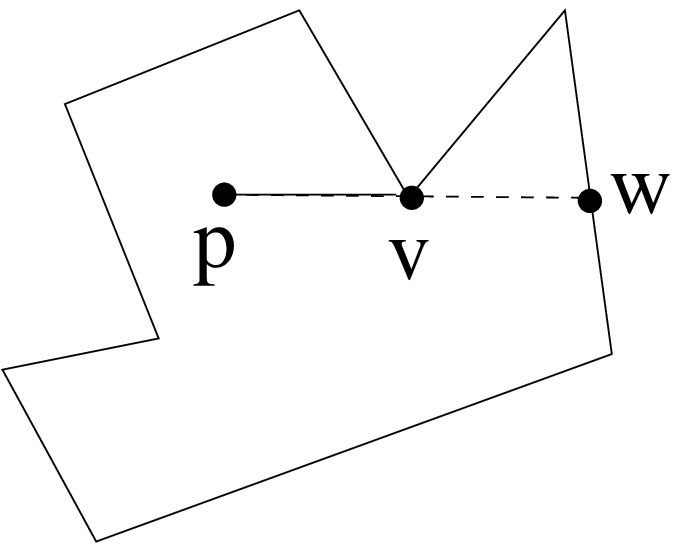}
\caption{\footnotesize Illustrating a window $\overline{vw}$ of $p$.}
\label{fig:window}
\end{center}
\end{minipage}
\hspace{0.02in}
\begin{minipage}[t]{0.49\linewidth}
\begin{center}
\includegraphics[totalheight=0.8in]{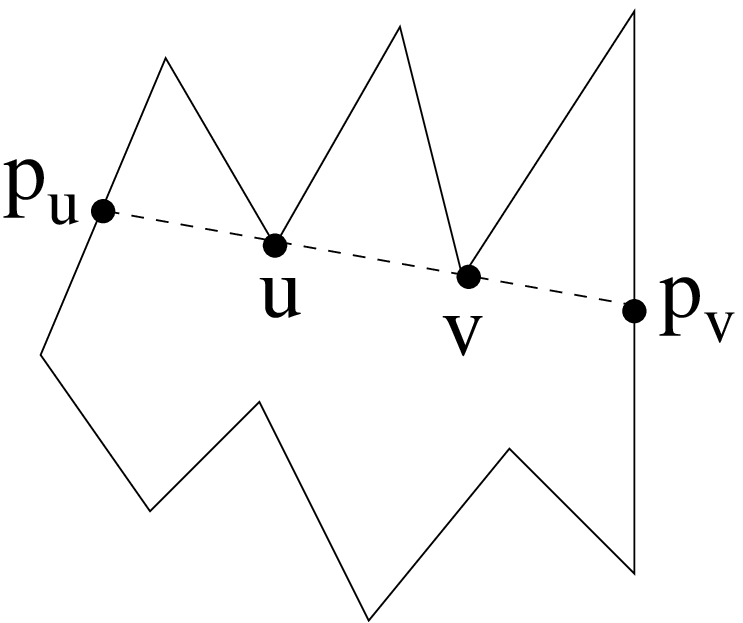}
\caption{\footnotesize Illustrating the two critical constraints
$\overline{vp_v}$ and $\overline{up_u}$ defined by
the two mutually visible vertices $u$ and $v$.}
\label{fig:subproblem}
\end{center}
\end{minipage}
\vspace*{-0.15in}
\end{figure}

As in \cite{ref:AronovVi02,ref:BoseEf02}, we represent the
visibility polygon $Vis(s)$ of a segment $s$ by a cyclic list of the
vertices and edges of $\calP$ in the order in which they appear on the
boundary of $Vis(s)$, and we call such a list the {\em combinatorial
representation} of $Vis(s)$ \cite{ref:AronovVi02}. With the
combinatorial representation, $Vis(s)$ can be explicitly determined in
linear time in terms of the size of $Vis(s)$.
Our query algorithms given later always report the combinatorial
representation of $Vis(s)$.

The critical constraints of $\calP$ partition $\calP$
into cells, called the {\em visibility decomposition} of $\calP$
and denoted by $\VD(\calP)$. The visibility decomposition $\VD(\calP)$ has a
property that for any two points $p$ and $q$ in the same cell of $\VD(\calP)$, the two
visibility polygons $Vis(p)$ and $Vis(q)$ have the same combinatorial
representation. Also, the combinatorial representations of the
visibility polygons of two adjacent cells in $\VD(\calP)$ have only
$O(1)$ differences.  The visibility decomposition has been used for
computing visibility polygons of query points (not line segments)
\cite{ref:AronovVi02,ref:BoseEf02}.

Consider a query segment $s$ in $\calP$. In the following, we discuss
an algorithmic scheme for computing $Vis(s)$.
Denote by $a$ and $b$ the two endpoints of $s$. Suppose we move a
point $p$ on $s$ from $a$ to $b$. We want to capture the combinatorial
representation changes of $Vis(p)$ of the point
$p$ during its movement.
Initially, $p$ is at $a$ and
we have $Vis(p) = Vis(a)$. As $p$ moves, the combinatorial representation of
$Vis(p)$ changes if and only if $p$ crosses a critical constraint of
$\calP$ \cite{ref:AronovVi02,ref:BoseEf02}.
$Vis(s)$ is the union of all such visibility polygons
as $p$ moves from $a$ to $b$. Therefore, to compute $Vis(s)$, as in
\cite{ref:AronovVi02,ref:BoseEf02}, we use the following approach.
Initially, let $Vis(s)=Vis(p)=Vis(a)$. As $p$ moves from $a$ to
$b$, when $p$ crosses a critical constraint, either $p$ sees one more
vertex/edge, or $p$ sees one less vertex/edge. If $p$ sees one more
vertex/edge, then we update $Vis(s)$ in constant time by inserting the new
vertex/edge
to the appropriate position of the combinatorial representation of
$Vis(s)$.
Otherwise, we do nothing (because even though a vertex/edge is
not visible to $p$ any more, it is visible to $s$ and thus should be
kept; refer to \cite{ref:BoseEf02} for details).

The above algorithm has two
remaining issues. The first one is how to compute $Vis(a)$ of the
point $a$. The second issue is how to determine the next critical
constraint that will be crossed by $p$ as $p$ moves.
Each of our two data structures given in Sections \ref{sec:first} and
\ref{sec:second} does some preprocessing such that the corresponding
query algorithm can resolve the above two issues efficiently.

\subsection{The Ray-Rotating Queries}
\label{sec:rayrotate}

Our first data structure in Section \ref{sec:first} needs the
following {\em ray-rotating} queries. Given any ray $\rho$ whose
origin $z$ is in $\calP$, the ray-rotating query asks for
the first vertex of $\calP$ visible to $z$ that will be hit by
$\rho$ when we rotate $\rho$ clockwise (or counterclockwise) around
$z$ (e.g., see Fig.~\ref{fig:rayrotate}).
By making use of the ray-shooting data structures
\cite{ref:ChazelleRa94,ref:ChazelleVi89,ref:GuibasLi87,ref:HershbergerA95}
and the two-point shortest path query data structure
\cite{ref:GuibasOp89}, we obtain the following result.

\begin{figure}[t]
\begin{minipage}[t]{0.49\linewidth}
\begin{center}
\includegraphics[totalheight=1.0in]{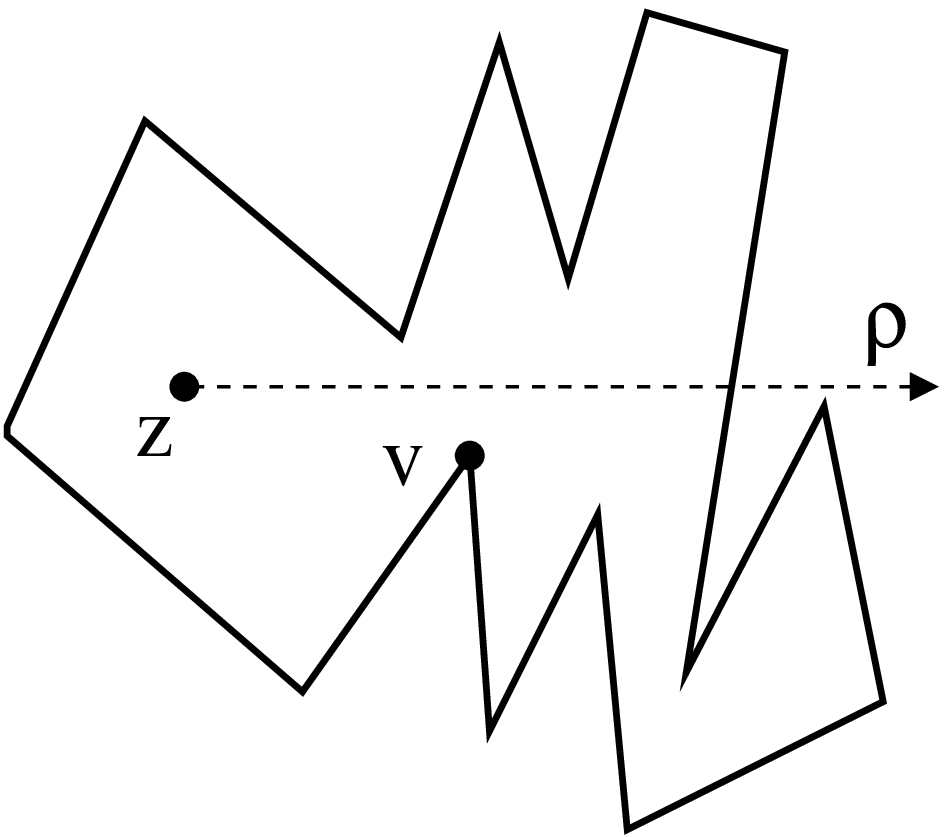}
\caption{\footnotesize Illustrating the ray-rotating query for $\rho$:
The vertex $v$ is the first visible vertex to $z$ that will be hit by
$\rho$ if we rotate $\rho$ clockwise around $z$.}
\label{fig:rayrotate}
\end{center}
\end{minipage}
\hspace*{0.06in}
\begin{minipage}[t]{0.49\linewidth}
\begin{center}
\includegraphics[totalheight=1.0in]{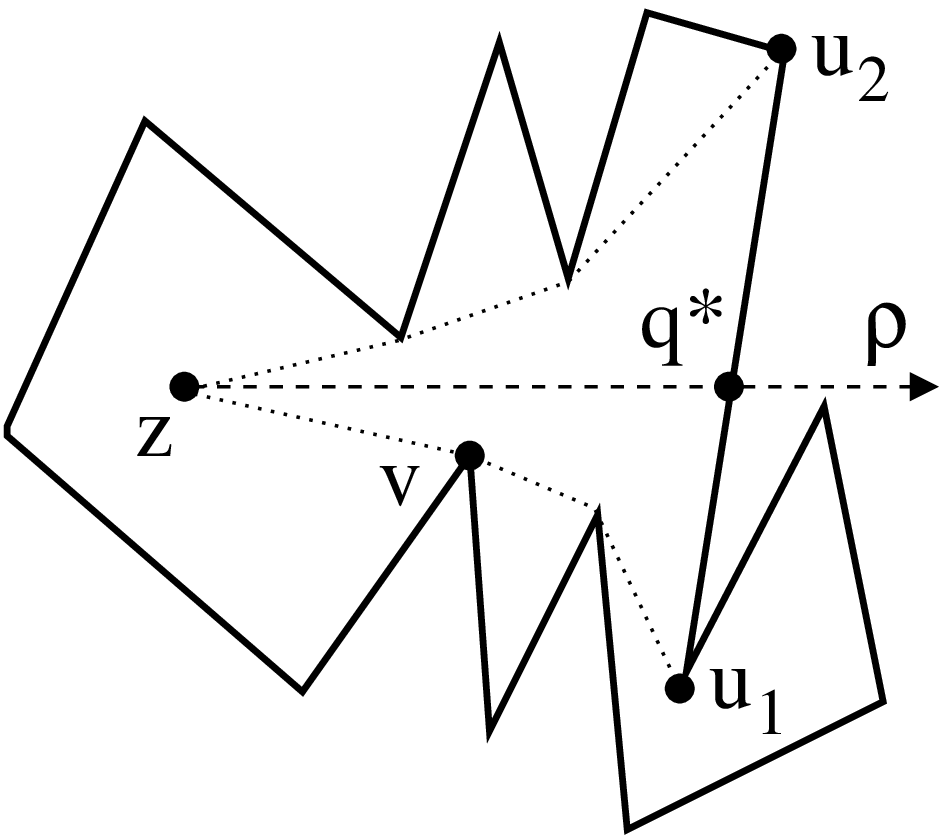}
\caption{\footnotesize Illustrating the proof for Lemma
\ref{lem:rayrotate}: The two dotted paths are shortest paths from $z$
to $u_1$ and $u_2$, respectively.}
\label{fig:funnel}
\end{center}
\end{minipage}
\vspace*{-0.15in}
\end{figure}

\begin{lemma}\label{lem:rayrotate}
A data structure can be built in $O(n)$ time and $O(n)$ space such
that each ray-rotating query can be answered in $O(\log n)$ time.
\end{lemma}
\begin{proof}
Consider any ray $\rho$ whose origin $z$ is in $\calP$.
Without
loss of generality, assume $\rho$ is horizontally rightwards.
Let $v^*$ be the sought vertex for the ray-rotating query on $\rho$,
i.e., $v^*$ is the first vertex of $\calP$ visible to $z$ that
will be hit by $\rho$ when we rotate $\rho$ clockwise around $z$ (the case of
counterclockwise rotation can be done similarly).

In the preprocessing, we compute a ray-shooting data structure
\cite{ref:ChazelleRa94,ref:ChazelleVi89,ref:GuibasLi87,ref:HershbergerA95}
in $O(n)$ time and space, such that given any ray with the origin in $\calP$, the first
point on the boundary of $\calP$ hit by the ray can be found in $O(\log n)$ time.
We also compute a two-point shortest path query data structure
\cite{ref:GuibasOp89} in $O(n)$ time and space, such that given any
two points $p$ and $q$ in $\calP$, the shortest path length between
$p$ and $q$ can be computed in $O(\log n)$ time and the path itself
can be found in additional time proportional to the number of turns
along it.

Our query algorithm for finding $v^*$ works as follows.

First, we use the ray-shooting data structure to find in $O(\log n)$ time
the first point
$q^*$ on the boundary of $\calP$ hit by $\rho$; the edge of
$\calP$ containing $q^*$ is also known immediately from the ray-shooting query.
If $q^*$ is a vertex of $\calP$, then $v^*=q^*$ and we are done;
otherwise, let the end
vertices of the edge of $\calP$ containing $q^*$ be $u_1$ and $u_2$
(e.g., see Fig.~\ref{fig:funnel}).
Let $\pi_1$ be the shortest path in $\calP$ from $z$ to $u_1$, and
similarly, let $\pi_2$ be the shortest path from $z$ to $u_2$.
Since $z$ is visible to $q^*$ on $\overline{u_1u_2}$, the
region bounded by $\pi_1$, $\pi_2$, and $\overline{u_1u_2}$ is a
funnel \cite{ref:GuibasLi87,ref:GuibasOp89,ref:LeeEu84}, with $z$ as
the apex (e.g., see Fig.~\ref{fig:funnel}).
Recall that $\rho$ is horizontally rightwards;
one vertex of $u_1$ and $u_2$ must be below the line
containing $\rho$. Without loss of generality, let $u_1$ be below the
line containing $\rho$.
Let $v$ be the vertex on $\pi_1$ that is connected to $z$ by a line
segment on $\pi_1$, i.e., $\overline{zv}$ is the first edge of $\pi_1$
(e.g., see Fig.~\ref{fig:funnel}).
Note that $v=u_1$ is possible, in which case $\pi_1$ is the line segment
$\overline{zu_1}$.
An easy observation is that the sought vertex $v^*$ is exactly the vertex $v$.
By using the two-point shortest path data structure
\cite{ref:GuibasOp89} on $z$ and $u_1$, the vertex $v$ can be easily
found in $O(\log n)$ time because $\overline{zv}$ is the first edge of
$\pi_1$.

Therefore, the sought vertex $v^*$ for the ray-rotating query on
$\rho$ can be found in $O(\log n)$ time.
The lemma thus follows.
\end{proof}

%\sectionspace
\section{The First Data Structure}
\label{sec:first}
%\subsectionspace

Our goal is to compute $Vis(s)$
for any query segment $s$. Again, let $s=\overline{ab}$.
As discussed before, we need to resolve two issues.  The first issue is to
compute $Vis(a)$.  For this, as discussed in
\cite{ref:AronovVi02}, by using the ray-shooting data structure
\cite{ref:ChazelleRa94,ref:ChazelleVi89,ref:GuibasLi87,ref:HershbergerA95},
with $O(n)$ time preprocessing, we
can compute $Vis(a)$ in $O(|Vis(a)|\log n)$ time, where $|Vis(a)|$ is
the size of $Vis(a)$. Note that it might be easier to compute
$|Vis(a)|$ by using both the ray-shooting data structure and our
ray-rotating data structure in Lemma \ref{lem:rayrotate}.

The second issue is how to determine the next critical constraint
of $\calP$ that will be crossed by the point $p$ as $p$ moves from $a$
to $b$. Suppose at the moment we know $Vis(p)$ (initially,
$Vis(p)=Vis(a)$). Let
$\beta$ be the critical constraint that is crossed next by $p$.
To determine $\beta$, we first sketch an observation given in
\cite{ref:AronovVi02}.

Denote by $T(p)$ the shortest path tree rooted at $p$, which is the
union of the shortest paths in $\calP$ from $p$ to all vertices of $\calP$. A
vertex of $\calP$ is in $Vis(p)$ if and only if it is a child of $p$ in $T(p)$.
For any child $v$ of $p$ in the tree $T(p)$,
define the {\em principal
child} of $v$ to be the vertex $w$ among the children of $v$ in $T(p)$ such that
the angle formed by the rays $\overrightarrow{vw}$ and $\overrightarrow{pv}$ is
the smallest among all such angles (see Fig.~\ref{fig:principal}).
In other words, if we go from $p$
to any child of $v$ along the shortest path and we turn to the left
(resp., right), then $w$ is the first child of $v$ that is hit by
rotating counterclockwise (resp., clockwise) the line containing
$\overline{pv}$ around $v$.

To determine the next critical constraint $\beta$, the
following observation was shown in \cite{ref:AronovVi02}. Two children
of $p$ in $T(p)$ are {\em consecutive} if there is no other child of
$p$ between them in the cyclic order around $p$.

\begin{figure}[t]
\begin{minipage}[t]{\linewidth}
\begin{center}
\includegraphics[totalheight=1.0in]{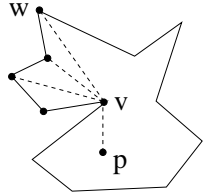}
\caption{\footnotesize Illustrating the principle child $w$ of $v$ in
$T(p)$.}
\label{fig:principal}
\end{center}
\end{minipage}
\vspace*{-0.15in}
\end{figure}

\begin{observation}\label{obser:10}{\em \cite{ref:AronovVi02}}
The next critical constraint $\beta$ is defined by two vertices of
$\calP$ that are either two consecutive children of $p$ or one, say
$v$, is a child of $p$ and the other is the principal child of $v$.
\end{observation}

Based on Observation \ref{obser:10}, Aronov {\em et al.}
\cite{ref:AronovVi02} maintained $T(p)$ as $p$
moves and used the balanced triangulation of $\calP$ to determine the
principal children.
Their query algorithm takes $O(k\log^2 n)$ time, where
$k=|Vis(s)|$, and the preprocessing time and space of their data structure
\cite{ref:AronovVi02} are $O(n^2\log n)$ and $O(n^2)$, respectively.

Here, we take a different approach, although we still use
Observation \ref{obser:10}. Our data structure consists of
the following: the ray-shooting data structure
\cite{ref:ChazelleRa94,ref:ChazelleVi89,ref:GuibasLi87,ref:HershbergerA95},
the ray-rotating data structure in Lemma \ref{lem:rayrotate},
and a priority queue $Q$.

We assume $Vis(a)$ has been computed.
We use the ray-rotating data structure to determine the
principal children of all children of $p$, as follows. First of all,
since we already know $Vis(p)$ (initially $p=a$),
we have all $p$'s children in $T(p)$, sorted cyclically
around $p$. Note that we need not store the entire tree $T(p)$.
Consider any child $v$ of $p$ in $T(p)$ (i.e., $v$ is visible to $p$).
%the vertices of $\calP$ that are visible to
%$v$ are those adjacent vertices of $v$ in the visibility graph $G$.
%(For simplicity, vertices of $G$ also refer to the corresponding
%vertices of $\calP$, and vice versa.)
%Denote by $adj(v)$ the set of the vertices of $\calP$ visible to $v$.
Let $w$ be
the principal child of $v$ that we are looking for.
%Clearly, $w\in adj(v)$.
Consider the ray $\rho(v)$
originating from $v$ with the direction from $p$ to $v$. By the
definition of principle children, $w$ is the first vertex of $\calP$
visible to $v$ that will be hit by the ray
$\rho(v)$ if we rotate $\rho(v)$ around $v$ along the direction that is
consistent
with the turning direction of the shortest paths from $p$ to the
children of $v$ in $T(p)$. It is easy to see that once we know the
above rotation direction, we can obtain $w$ in $O(\log n)$ time by our
ray-rotating data structure in Lemma \ref{lem:rayrotate}.

To determine the
above rotation direction, we only need to look at the relationship between the
line containing $\rho(v)$ and the two edges of $\calP$ adjacent to
$v$. Specifically, assume the line containing $\rho(v)$ has the same
direction as $\rho(v)$. For example, if the two adjacent edges of $v$ both lie to the
left of this line (e.g., see Fig.~\ref{fig:principal}), then we
should rotate $\rho(v)$ counterclockwise to determine $w$.
%otherwise, we rotate $\rho(v)$ clockwise.
The other cases can be determined in a similar manner.
In summary, we can obtain the principle
child of $v$ in $O(\log n)$ time.  Initially, $p=a$ and
we determine the principle children of all children of $a$ in
$T(a)$ in $O(|Vis(a)|\log n)$ time since $a$ has $O(|Vis(a)|)$ children in $T(a)$.

We use the priority queue $Q$ to store the critical constraints specified in
Observation \ref{obser:10} that intersect the line segment $s$, where
the key of each such critical constraint used in the priority queue $Q$ is the position of its
intersection with $s$. Initially when $p=a$, we compute the critical constraints
defined by all pairs of consecutive children of $p$ in $T(p)$.
Similarly, for each child $v$ of
$T(p)$, we compute the critical constraint defined by $v$ and its
principal child.  Note that the total number of these critical
constraints is $O(|Vis(a)|)$. For each such critical constraint, we
check whether it intersects $s$, which can be done in $O(\log n)$ time
with the help of a ray-shooting query (we omit the details).
If the critical constraint intersects $s$, we insert it into $Q$;
otherwise, we do nothing. Then, the first critical constraint in $Q$ is
the next critical constraint that $p$ will cross as it moves. In
general, after $p$ crosses a critical constraint, $p$ either sees one
more vertex or sees one less vertex of $\calP$. In either case, there are only a
constant number of insertion or deletion operations on $Q$.
Specifically, consider the case when $p$ sees one more vertex $u$
(and an adjacent edge of $u$). By
the implementation given in \cite{ref:BoseEf02}, we can update the
combinatorial representation of $Vis(p)$ in constant time (i.e.,
insert $u$ and the adjacent edge to the appropriate positions of
the cyclic list of $Vis(p)$).  After this,
$u$ becomes a child of $p$ in the new tree $T(p)$, and
we can determine $p$'s two other children, say, $u_1$ and $u_2$, which
are cyclic neighbors of $u$, in constant time.
Then, for $u_1$, we check whether the critical constraint
defined by $u$ and $u_1$ intersects $s$, and if so, we insert it into
$Q$. For $u_2$, we do the same thing. Further, we compute the
principal child of $u$ in $T(p)$, in $O(\log n)$ time, by the approach
discussed above. For the other case where $p$ sees one less vertex
after it crosses the critical constraint, we perform similar processing.
After $p$ arrives at the other endpoint $b$ of $s$, we obtain the
combinatorial representation of $Vis(s)$.

We claim that the above algorithm takes $O(k\log
n)$ time (with $k=|Vis(s)|$).
Indeed, the initialization takes $O(|Vis(a)|\log n)$ time.
Clearly, $|Vis(a)|=O(k)$ since each vertex of $\calP$ that is in
$Vis(a)$ also appears in $Vis(s)$. If we consider every time when $p$
crosses a critical constraint as an {\em event}, then each event takes
$O(\log n)$ time. It has been shown in \cite{ref:AronovVi02} that the
total number of events as $p$ moves from $a$ to $b$ is $O(k)$. Hence,
the overall running time for computing $Vis(s)$ is $O(k\log n)$.

For the preprocessing, the ray-shooting data structure and the
ray-rotating data structure both need $O(n)$
time and space to build. Further, in our query algorithm, the space used in the
priority queue $Q$ is always bounded by $O(k)$. We conclude
this section with the following result.

\begin{theorem}\label{theo:10}
For any simple polygon $\calP$, a data structure can be built in
$O(n)$ time and $O(n)$ space, such that the visibility polygon
$Vis(s)$ can be computed in $O(|Vis(s)|\log n)$ time for any query
line segment $s$ in $\calP$.
\end{theorem}

\section{The Second Data Structure}
\label{sec:second}

In general, the preprocessing of our second data structure is very similar
to that in \cite{ref:BoseEf02}, and we make it faster by using better tools.
Our improvement on the query algorithm
is based on a number of new observations, e.g., a combinatorial bound of the
``zone" of the line segment arrangements in simple polygons.
For completeness, we first briefly discuss the approach in \cite{ref:BoseEf02}.

The preprocessing in \cite{ref:BoseEf02} has several steps, whose
running time is $O(n^3\log n)$ and is dominated by the first
two steps. The other steps together take $O(n^3)$ time. We show
below that the first two steps can be implemented in $O(n^3)$ time.

The preprocessing in \cite{ref:BoseEf02} first computes the visibility
decomposition $\VD(\calP)$ of $\calP$. Although there may be
$\Omega(n^2)$ critical constraints in $\calP$, it has been shown
\cite{ref:BoseEf02} that any line segment in $\calP$ can intersect
only $O(n)$ critical constraints, which implies that the size of
$\VD(\calP)$ is $O(n^3)$ instead of $O(n^4)$.
All critical constraints of $\calP$ can be
computed in $O(n^2)$ time, e.g., by the algorithm in \cite{ref:GhoshAn91}.
After that, to compute $\VD(\calP)$, we can use
%the authors in \cite{ref:BoseEf02} used the algorithm in
%\cite{ref:BentleyAl79}, which takes $O(n^3\log n)$ time. In fact, a
%faster algorithm is available for computing $\VD(\calP)$.
Chazelle and Edelsbrunner's algorithm \cite{ref:ChazelleAn92Edelsbrunner}, which
computes the planar
subdivision induced by a set of $m$ line segments in $O(m\log m+I)$
time, where $I$ is the number of intersections of the line segments. In
our problem, we have $O(n^2)$ critical constraints each of which is a
line segment and the boundary of $\calP$ has $n$ edges. Therefore, by
using the algorithm in \cite{ref:ChazelleAn92Edelsbrunner}, we
can compute $\VD(\calP)$ in $O(n^3)$ time.
Alternatively, an approach mentioned in \cite{ref:HershbergerA95} can also be
used to compute $\VD(\calP)$ in $O(n^3)$ time, and we omit the details.

%The authors
%\cite{ref:BoseEf02} mentioned that the preliminary version of Chazelle
%and Edelsbrunner's algorithm \cite{ref:ChazelleAn88Edelsbrunner} requires that the
%line segments be in general position and the authors
%\cite{ref:BoseEf02} were not clear how to guarantee this. Here, we
%note that the full version of Chazelle and Edelsbrunner's algorithm
%\cite{ref:ChazelleAn92Edelsbrunner} could also handle the general case.

The second step of the preprocessing in \cite{ref:BoseEf02} is to
build a planar point location data structure on $\VD(\calP)$ in $O(n^3\log
n)$ time. By the approaches in \cite{ref:EdelsbrunnerOp86}
or \cite{ref:KirkpatrickOp83}, we can build such a point location data
structure in $O(n^3)$ time.

The remaining steps of our preprocessing algorithm are the same as those in
\cite{ref:BoseEf02},
which together take $O(n^3)$ time. Hence, the total preprocessing time is
$O(n^3)$. With the preprocessing, for each query point $q$ in $\calP$,
we can compute the visibility polygon $Vis(q)$ of $q$ in
$O(|Vis(q)|+\log n)$ time.

For a query segment $s=\overline{ab}$, the query algorithm in
\cite{ref:BoseEf02} first computes $Vis(a)$. Then, again, let a point
$p$ move on $s$ from $a$ to $b$. The algorithm maintains $Vis(p)$ as
$p$ moves on $s$, initially with $Vis(p)=Vis(a)$. Again, whenever $p$ crosses
a critical constraint, the combinatorial representation of $Vis(p)$
changes. Unlike our first data structure in Section \ref{sec:first},
here we have $\VD(\calP)$ explicitly. Therefore, we can determine the
next critical constraint in a much easier way. Specifically, the algorithm in
\cite{ref:BoseEf02} uses the following approach. Suppose $p$ is
currently in a cell of $\VD(\calP)$; then the next critical constraint
crossed by $p$ must be on the boundary of that cell. Since each cell is
convex, we can determine this critical constraint in $O(\log n)$
time. The algorithm stops when $p$ arrives at $b$. The total running
time of the query algorithm is $O(k\log n)$, where $k=|Vis(s)|$.

We propose a new and simpler query algorithm. We follow the
previous query algorithmic scheme. The only difference is when we determine
the next critical constraint that will be crossed by $p$, we simply
check each edge on the boundary of the current cell that contains $p$, and the running time is
linear in terms of the number of edges of the cell. Therefore, the
total running time of finding all critical constraints crossed by $p$
as it moves on $s$ is proportional to the total number of edges on all
faces of $\VD(\calP)$ that intersect $s$, and we denote by $F(s)$ the set of
such faces of $\VD(\calP)$. Let $E(s)$ denote the set of
edges of the faces in $F(s)$. Then the total time of finding
all critical constraints crossed by $p$ is $O(|E(s)|)$. Note that the
time of the overall query algorithm is the sum of the time for
computing $Vis(a)$ and the time for finding
all critical constraints crossed by $p$. Since $Vis(a)$
can be found in $O(|Vis(a)|+\log n)$ time, the running time of the query
algorithm is $O(\log n+|Vis(a)|+|E(s)|)$. Recall that $|Vis(a)|=O(k)$.
In Lemma~\ref{lem:10} below,
we will prove $|E(s)|=O(k)$. Consequently, the query
algorithm takes $O(\log n+k)$ time and Theorem
\ref{theo:20} below thus follows.

\begin{lemma}\label{lem:10}
The size of the set $E(s)$ is $O(k)$.
\end{lemma}

\begin{theorem}\label{theo:20}
For any simple polygon $\calP$, we can build a data structure of
size $O(n^3)$ in $O(n^3)$ time that can compute $Vis(s)$ in
$O(|Vis(s)|+\log n)$ time for each query segment $s$ in $\calP$.
\end{theorem}

\subsection{Proving Lemma \ref{lem:10}}

It remains to prove Lemma \ref{lem:10}.
Note that each edge of $E(s)$
lies either on $\partial\calP$ or on a critical constraint.
We partition the set $E(s)$ into two subsets $E_1(s)$ and $E_2(s)$.
For each edge of $E(s)$, if it lies on $\partial\calP$, then
it is in $E_1(s)$; otherwise, it is in $E_2(s)$. We will show that
both $|E_1(s)|=O(k)$ and $|E_2(s)|=O(k)$ hold.

Denote by $C(s)$ the set of all critical constraints each of which
contains at least one edge of $E(s)$.
%The proof of the following lemma is in the appendix.

\begin{lemma}\label{lem:20}
The size of the set $C(s)$ is $O(k)$.
\end{lemma}
\begin{proof}
Denote by
$V(s)$ the set of vertices of $\calP$ visible to $s$. Clearly,
$|V(s)|\leq k$. Consider an arbitrary critical constraint $c\in C(s)$.
To prove the lemma, we will charge $c$ to a vertex of $V(s)$.
We will show that each vertex of $V(s)$ will be charged at most
a constant number of times, which will lead to the lemma.

Assume the defining vertex of $c$ is $u$ and the anchor vertex
of $c$ is $v$. By
the definition of $C(s)$, $c$ contains at least one edge of $E(s)$.
Depending on whether $c$ intersects $s$, there are two cases.

\begin{enumerate}
\item
If $c$ intersects $s$, then the defining
vertex $u$ of $c$ must be visible to $s$ (see Fig.~\ref{fig:intersect}).
To see this, let $q$ be the
intersection of $c$ and $s$. Hence, $q$ is visible to $v$. One may
consider that the visibility between $q$ and $u$ is blocked by $v$. Due to
our assumption that each endpoint of $s$ is not collinear with two
vertices of $\calP$, $q$ is not an endpoint of $s$. Hence,
there is always a point on $s$ infinitely close to $q$ that is visible
to $u$, and thus $u$ is visible to $s$. We charge $c$ to its defining vertex $u$.

\begin{figure}[t]
\begin{minipage}[t]{0.49\linewidth}
\begin{center}
\includegraphics[totalheight=1.0in]{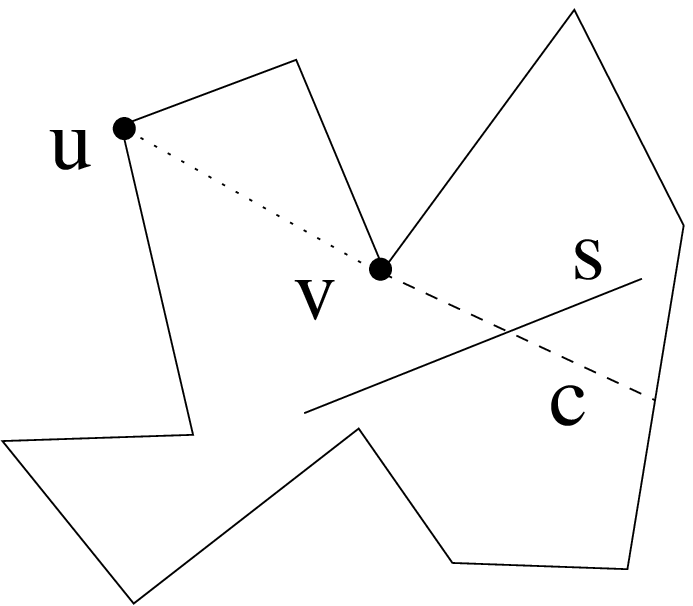}
\caption{\footnotesize Illustrating the case where $c$ intersects $s$.}
\label{fig:intersect}
\end{center}
\end{minipage}
\hspace{0.06in}
\begin{minipage}[t]{0.49\linewidth}
\begin{center}
\includegraphics[totalheight=1.0in]{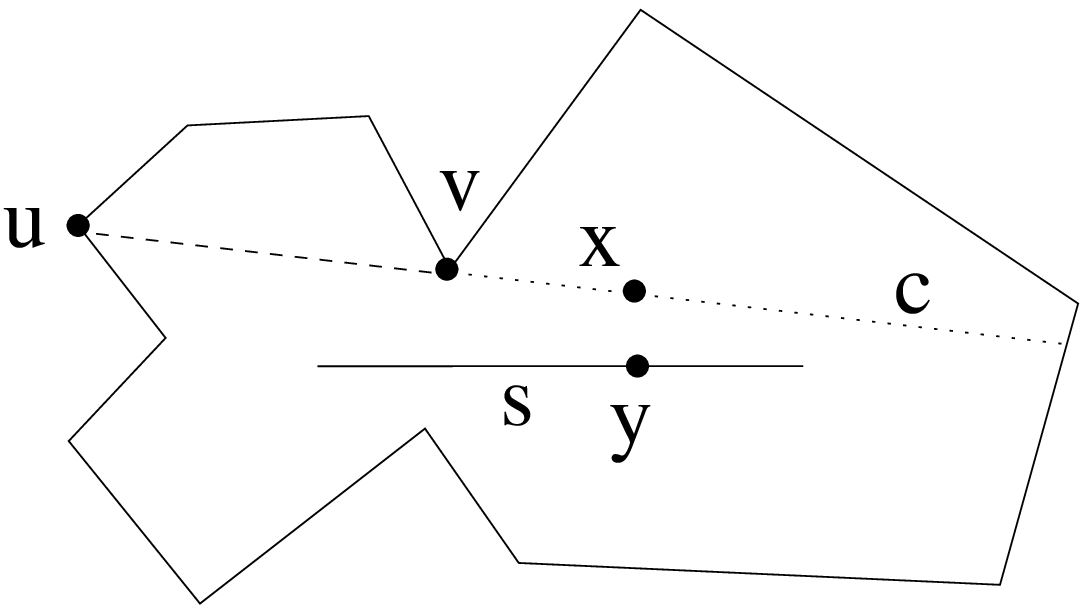}
\caption{\footnotesize Illustrating the case where $c$ does not intersect $s$.}
\label{fig:nonintersect}
\end{center}
\end{minipage}
\vspace*{-0.15in}
\end{figure}

\item

If $c$ does not intersect $s$ (see Fig.~\ref{fig:nonintersect}),
then we show below that the anchor vertex $v$ of $c$ must be visible
to $s$, and further, there are at most two critical constraints in
$C(s)$ such that they do not intersect $s$ and their
anchor vertices are $v$. We will charge $c$ to $v$.

We first
prove that $v$ is visible to $s$.
Since $c$ contains at least one edge of $E(s)$, there must be a face
$f$ of $\VD(\calP)$ intersecting $s$ and the
boundary of $f$ has an edge $e$ contained in $c$.
Let $x$ be an arbitrary interior point of $e$ and let $y$ be an arbitrary point
on $s$ that is contained in $f$ (see Fig.~\ref{fig:nonintersect}).
Since $f$ is convex,
$\overline{xy}$ is contained in $f$, i.e., $x$ is visible to $y$
and $\overline{xy}$ does not intersect any other critical constraint
of $\calP$ than $c$ (at $x$). To prove $y$ is visible to the vertex $v$,
consider a point $q$ on $\overline{xy}$ moving from $x$ to $y$. We
claim that $v$ is always visible to $q$ as $q$ moves. Indeed,
initially $q$ is at $x$, and $v$ is visible to $x$ because $x$ is on
the critical constraint $c$ and $v$ is the anchor vertex of $c$.
Suppose to the contrary $v$ is not visible to $q$ at some moment as $q$ moves.
Then, at some moment, $\overline{vq}$ must encounter a vertex of
$\calP$, say, $w$. In other words, $w$ is on $\overline{vq}$. Then,
the two vertices $v$ and $w$ define a critical constraint with $v$ as
the defining vertex and $w$ as the anchor vertex, and the critical
constraint intersects $\overline{xy}$ at $q$. Note that this
critical constraint is not $c$ because $v$ is the anchor vertex of
$c$. Hence, we obtain a contradiction because $c$ is the only critical
constraint that intersects $\overline{xy}$. Therefore, we conclude that
$v$ is visible to $y$.

Next, we prove that there are at most two critical constraints in
$C(s)$ such that they do not intersect $s$ and their
anchor vertices are $v$. Let $C_v$ denote the set of all critical constraints
each of which has $v$ as its anchor vertex and does not intersect $s$.
Our goal is to prove that $C_v$ has at most two critical constraints in
$C(s)$.
Note that each critical constraint in $C_v$ has $v$ as an endpoint and
its other endpoint is on $\partial \calP$. Hence, the critical
constraints of $C_v$ partition $\calP$ into $|C_v|+1$
interior-disjoint regions
and one region contains $s$ (see Fig.~\ref{fig:partition}). Let $R(s)$ be
the region containing $s$. Clearly, $R(s)$ has at most two critical
constraints of $C_v$, say $c_1$ and $c_2$, on its boundary.  We claim
that for any critical constraint $c'\in C_v\setminus\{c_1,c_2\}$, $c'$
cannot contain an edge of $E(s)$. Indeed, assume to the contrary $c'$
contains an edge of $E(s)$. Then, as discussed before, we can always
find such two points $x$ and $y$ as in Fig.~\ref{fig:nonintersect}.
Recall that $\overline{xy}$ is in $\calP$ and $\overline{xy}$ does not
intersect any other critical constraint of $\calP$ than $c'$.
Since $c'$ is outside $R(s)$, $x\in c'$ is outside $R(s)$. However,
due to $y\in s$ and $s\subset R(s)$, $\overline{xy}$ must intersect either
$c_1$ or $c_2$, which contradicts with that $\overline{xy}$ does not
intersect any other critical constraint of $\calP$ than $c'$.
Hence, $c'$ cannot contain an edge of $E(s)$ and $c'\not\in
C(s)$. Therefore, we obtain that
$C_v$ has at most two critical constraints in $C(s)$.

\begin{figure}[t]
\begin{minipage}[t]{\linewidth}
\begin{center}
\includegraphics[totalheight=1.0in]{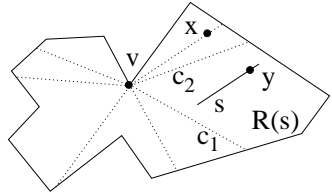}
\caption{\footnotesize Illustrating the critical constraints (the
dotted line segments) with $v$ as
their anchor vertex that do not intersect $s$.}
\label{fig:partition}
\end{center}
\end{minipage}
\vspace*{-0.15in}
\end{figure}
\end{enumerate}

According to our discussion above, in the first case (i.e., $c$
intersects $s$), we charge $c$ to
its defining vertex $u$, which is in $V(s)$. In the second case (i.e., $c$
does not intersect $s$), we charge $c$ to its anchor vertex $v$, which
is also in $V(s)$.
An observation in \cite{ref:BoseEf02} shows that for any line segment in
$\calP$, for any vertex $u$ of $\calP$, the line segment intersects
at most two critical constraints with $u$ as their defining vertex.
Therefore, for any vertex $u$ of $\calP$,
$u$ can be charged at most twice as a defining vertex.
On the other hand, we have shown that, as an anchor vertex,
$v$ has at most two critical constraints in $C(s)$ that do not
intersect $c$. Therefore, for any vertex $v$ of $\calP$,
$v$ can be charged at most twice as an anchor vertex.
Hence, any vertex in $V(s)$ can be charged at most four times, twice
as an anchor vertex and twice as a defining vertex. In other words,
$|C(s)|\leq 4\cdot |V(s)|\leq 4\cdot k$.

The lemma thus follows.
\end{proof}

In the next lemma, we bound the size of the subset $E_1(s)$.

\begin{lemma}\label{lem:30}
The size of the set $E_1(s)$ is $O(k)$.
\end{lemma}
\begin{proof}
Denote by $V(s)$ the set of vertices
of $\calP$ visible to $s$. Clearly, $|V(s)|\leq k$.
Consider an edge $e$ in $E_1(s)$.  To prove the
lemma, we will charge $e$ either to a vertex of $V(s)$ or to a
critical constraint of $C(s)$. We will also show that each vertex of
$V(s)$ will be charged at most twice and each critical constraint of
$C(s)$ will be charged at most four times. Consequently, due to
$|V(s)|\leq k$ and $|C(s)|=O(k)$ (by Lemma \ref{lem:20}), the lemma follows.

By the definition of $E_1(s)$, $e$ is on an edge of $\calP$. If $e$
has an endpoint that is a vertex of $\calP$, say $u$, then clearly
$u$ is visible to $s$. We charge $e$ to $u$. Otherwise, both endpoints
of $e$ are endpoints of some critical constraints, and we charge $e$ to an
arbitrary one of the two such critical constraints.

For each vertex of $\calP$, it has two adjacent edges in $\calP$, and
therefore, it has at most two adjacent edges in $E_1(s)$. Hence, each
vertex of $V(s)$ can be charged at most twice. On the other hand,
each critical constraint has two endpoints, and each endpoint is
adjacent to at most two edges in $E_1(s)$. Therefore, each critical
constraint of $C(s)$ can be charged at most four times.
\end{proof}

To prove Lemma \ref{lem:10}, it remains to show $|E_2(s)|=O(k)$. To
this end, we discuss a more general problem, in the following.

Assume we have a set $S$ of line segments in $\calP$ such that the endpoints of
each such segment are on $\partial\calP$.
Let $\calA$ be the arrangement formed by the line segments of $S$ and
the edges of $\partial\calP$.
For any line segment $s$ in $\calP$ (the endpoints of $s$ need
not be on $\partial\calP$), the {\em zone} of $s$ is defined to be
the set of all faces of $\calA$ that $s$ intersects. Denote by
$Z(s)$ the zone of $s$. For each edge of a face in $\calA$, it
either lies on a line segment of $S$ or lies on $\partial\calP$; if it
is the former case, we call the edge an {\em $S$-edge}.
We define the {\em complexity} of $Z(s)$ as
the number of $S$-edges of the faces in $Z(s)$ (namely,
the edges on $\partial\calP$ are not considered), denoted by $\Lambda$.
Our goal is to find a good upper bound for $\Lambda$.
By using the zone theorem for the general line segment
arrangement \cite{ref:EdelsbrunnerAr92}, we can easily obtain $\Lambda=O(|S|\alpha(|S|))$,
where $\alpha(\cdot)$ is the functional inverse
of Ackermann's function \cite{ref:HartNo86}.

Denote by $S_s$ the set of line segments in $S$ that intersect
$Z(s)$, i.e., each segment in $S_s$ contains at least one $S$-edge of
$Z(s)$. Let $m=|S_s|$ (note that $m\leq |S|$).  By using the property
that each segment in $S$ has both endpoints on $\partial\calP$, we show $\Lambda=O(m)$
in Theorem \ref{theo:zone} below, which we call the {\em zone theorem}.
The proof of  Theorem \ref{theo:zone} is given in Section \ref{sec:zone}.

\begin{theorem}\label{theo:zone}
The complexity of $Z(s)$ is $O(m)$.
\end{theorem}

Now consider
our original problem of proving $|E_2(s)|=O(k)$. By using the zone theorem, we have the following corollary.

\begin{corollary}\label{cor:10}
The size of the set $E_2(s)$ is $O(k)$.
\end{corollary}
\begin{proof}
The set $E_2(s)$ consists of all edges of
$E(s)$ that lie on the critical constraints. Recall that each
critical constraint is a line segment in $\calP$ with both endpoints on
$\partial\calP$. Consider the arrangement formed by all critical
constraints of $\calP$ and $\partial\calP$. The complexity of the
zone $Z(s)$ of the query segment $s$ in this arrangement is exactly $|E_2(s)|$.
Let $C'(s)$ be the set of critical constraints of $\calP$ each of which contains at least one edge in $E_2(s)$. Then, by the zone theorem (Theorem \ref{theo:zone}), we have $|E_2(s)|=O(|C'(s)|)$. Note that
$C'(s)\subseteq C(s)$. Due to $|C(s)|=O(k)$ (Lemma \ref{lem:20}), we have
$|E_2(s)|=O(k)$. The corollary thus follows.
\end{proof}

Lemma \ref{lem:30} and Corollary \ref{cor:10} together lead to Lemma
\ref{lem:10}.

\subsectionspace
\subsection{Proving the Zone Theorem (i.e., Theorem \ref{theo:zone})}
\label{sec:zone}

This subsection is devoted entirely to proving the zone theorem, i.e., Theorem \ref{theo:zone}.
All notations here are the same as defined before.

We partition the set $S_s$ into two subsets: $S_s^1$ and $S_s^2$.
For each segment in $S_s$, if it does not intersect the interior of
$s$, then it is in $S_s^1$; otherwise,
it is in $S_s^2$. Let $m_1=|S_s^1|$ and $m_2=|S_s^2|$.
Hence, $m=m_1+m_2$.  Consider the arrangement formed by the line segments
in $S_s^1$ and $\partial\calP$. Since no segment in $S_s^1$ intersects the
interior of $s$,
$s$ must be contained in a single face of the above arrangement and we denote
by $F_s$ that face. For each edge of $F_s$, if it lies on a segment of
$S$, we also call it an $S$-edge.
Note that the edges of $F_s$ that are not $S$-edges are all on $\partial\calP$.

\begin{lemma}\label{lem:40}
The number of $S$-edges of the face $F_s$ is $O(m_1)$; the shortest
path in $\calP$ between any two points in $F_s$ is contained in $F_s$.
\end{lemma}
\begin{proof}
For each segment $s'$ in $S_s^1$, since both endpoints of $s'$ are on
$\partial\calP$, $s'$ partitions $\calP$ into two simple polygons and
one of them contains $s$, which we denote by $\calP(s')$. It is easy
to see that the face $F_s$ is the common intersection of $\calP(s')$'s
for all $s'$ in $S_s^1$. To prove the lemma, it is sufficient to show that each
segment $s'$ in $S_s^1$ has at most one (maximal) continuous portion on the
boundary of $F_s$, as follows.

For any two points $p$ and $q$ in $\calP$, denote by $\pi(p,q)$ the
shortest path between $p$ and $q$ in $\calP$. Note that since $\calP$
is a simple polygon, $\pi(p,q)$ is unique. We claim that for any two
points $p$ and $q$ in the face $F_s$, $\pi(p,q)$ is contained in
$F_s$. Indeed, suppose to the contrary $\pi(p,q)$ is not
contained in $F_s$. Then, $\pi(p,q)$ must cross the boundary of $F_s$.
Since $\pi(p,q)$ cannot cross the boundary of $\calP$, $\pi(p,q)$ must
cross an $S$-edge of $F_s$, and we assume $s'$ is the segment in $S_s^1$
that contains such an $S$-edge. This implies that $\pi(p,q)$ is also not
contained in the polygon $\calP(s')$. Recall that the line segment
$s'$ partitions $\calP$ into two simple polygons and one of them is
$\calP(s')$. It is easy to show that for any two points in
$\calP(s')$, their shortest path in $\calP$ must be contained in
$\calP(s')$. Therefore, we obtain a contradiction. Hence, our above claim
is true.

Now assume to the contrary that a segment $s'$ in $S_s^1$ has two
disjoint maximal continuous portions on the boundary of $F_s$. Let $p$ and $q$ be
two points on these two portions of $s'$, respectively. Thus, both $p$
and $q$ are in $F_s$.
Since these are two discontinuous portions of $s'$ on the boundary of $F_s$, the line
segment $\overline{pq}$ is not contained in $F_s$. Since
$\overline{pq}$ is on $s'$, the
shortest path $\pi(p,q)$ is $\overline{pq}$. But this means $\pi(p,q)$
is not contained in $F_s$, which incurs a contradiction with our
previous claim that $\pi(p,q)$ must be contained in $F_s$.
Hence, we obtain that each
segment $s'$ in $S_s^1$ has at most one continuous portion on the
boundary of $F_s$, and consequently, the number of $S$-edges
of the face $F_s$ is $O(m_1)$.
\end{proof}

Lemma \ref{lem:50} below
shows a property of the face $F_s$.

\begin{lemma}\label{lem:50}
For any line segment $s'$ in $\calP$ with both endpoints on $\partial\calP$, $s'$ has
at most one (maximal) continuous portion intersecting $F_s$; consequently, $s'$ intersects
the interior of at most two edges of $F_s$.
\end{lemma}
\begin{proof}
Assume to the contrary that $s'$ has two disjoint maximal
continuous portions intersecting $F_s$. Let $p$ and $q$ be
two points on these two portions of $s'$, respectively. Thus, both $p$
and $q$ are in $F_s$. Clearly, the line
segment $\overline{pq}$ is not contained in $F_s$. Since
$\overline{pq}$ is on $s'$, $\overline{pq}$ is the
shortest path $\pi(p,q)$ between $p$ and $q$ in $\calP$. But this means $\pi(p,q)$
is not contained in $F_s$, which incurs a contradiction with
Lemma \ref{lem:40}. Hence, the lemma holds.
\end{proof}

%Next, we consider the segments in $S_s^2$.
For each $S$-edge of $Z(s)$, it lies either on a segment in $S_s^1$ or
on a segment in $S_s^2$; we call it an $S_s^1$-edge if it lies on a
segment in $S_s^1$ and an $S_s^2$-edge otherwise.
Due to $m=m_1+m_2$, our zone theorem is an immediate consequence of
Lemma \ref{lem:60} below. Note that we can obtain the zone $Z(s)$
of $s$ by adding the segments of $S_s^2$ to $F_s$.
To prove Lemma \ref{lem:60}, we use induction on $m_2$, i.e., $|S_s^2|$.
The approach is very similar to that in \cite{ref:deBergCo08}
used for line arrangements. Here, although we have line segments instead of lines, the
property that each line segment has both endpoints on $\partial\calP$
makes the approach in \cite{ref:deBergCo08} applicable with some modifications.
%The proof of Lemma \ref{lem:60} is given in the appendix.

\begin{lemma}\label{lem:60}
There are $O(m_2)$ $S_s^2$-edges
and $O(m_1+m_2)$ $S_s^1$-edges in the zone $Z(s)$.
\end{lemma}

\begin{proof}
Without loss of generality, assume the segment $s$ is horizontal.
 It is easy to see
that each $S_s^1$-edge bounds one face of $Z(s)$ and each $S_s^2$-edge
bounds two faces of $Z(s)$ (one lies on its right and the other lies
on its left). For each $S_s^2$-edge, we say it is a {\em left bounding
$S_s^2$-edge} for the face
lying on the right of it and a {\em right bounding $S_s^2$-edge} for the face
lying on the left of it.
Below we will prove that the number of left
bounding $S_s^2$-edges of the faces in $Z(s)$ is $O(m_2)$.
Analogously, the number of right bounding $S_s^2$-edges of the faces
in $Z(s)$ is also $O(m_2)$. In addition, we will also show that the
number of $S_s^1$-edges of $Z(s)$ is $O(m_1+m_2)=O(m)$.

Our proof is by induction on $m_2$. Consider the base case with $m_2=1$.
Denote by $s'$ the only segment in $S_s^2$.
Note that the face $F_s$ has no $S_s^2$-edges on its boundary and has $O(m_1)$
$S_s^1$-edges by Lemma \ref{lem:40}.
In light of Lemma \ref{lem:50}, $s'$ has at most one maximal continuous
portion intersecting the face $F_s$ and $s'$ intersects the interior of at most two
$S_s^1$-edges of $F_s$. Therefore, after we add $s'$ to $F_s$,
the number of $S_s^1$-edges of $Z(s)$ increases by at most two and the number of
left bounding $S_s^2$-edges increases by at most one.

\begin{figure}[h]
\begin{minipage}[t]{\linewidth}
\begin{center}
\includegraphics[totalheight=1.2in]{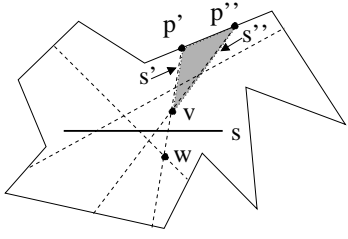}
\caption{\footnotesize The shaded region is $R(v)$, which is not in
the zone of $s$.}
\label{fig:arrangement}
\end{center}
\end{minipage}
\vspace*{-0.15in}
\end{figure}

Consider the general case of $m_2\geq 1$. Let $s'$ be the segment in $S_s^2$
that intersects $s$ at the rightmost position among all segments in $S_s^2$.
We first consider the case when this segment
$s'$ is unique. By induction, the zone of $s$ has $c\cdot (m_2-1)$ left
bounding $S_2^2$-edges and $c\cdot (m_1+m_2-1)$ $S_s^1$-edges, for
some constant $c$, without considering the segment $s'$.
Now consider adding $s'$.
First, by Lemma \ref{lem:50}, $s'$ has at most one maximal continuous
portion intersecting the face $F_s$ and $s'$ intersects the interior of at most two
$S_s^1$-edges of the zone $Z(s)$; therefore, the number of $S_s^1$-edges
increases by at most $2$. Second, the number of left bounding $S_s^2$-edges
increases in two ways: there are new left bounding $S_s^2$-edges on $s'$
and there are old left bounding $S_s^2$-edges that are split by $s'$. Let $v$
be the first intersection point of $s'$ with another segment in
$S_s^2$ above $s$, and let $w$ be the first intersection point of $s'$ with
another segment in $S_s^2$ below $s$ (e.g., see
Fig.~\ref{fig:arrangement}). We assume both $v$ and $w$ exist since otherwise the analysis is even simpler.
The segment $\overline{vw}$ on
$s'$ becomes a new left bounding $S_s^2$-edge. In addition, $s'$
splits a left bounding $S_s^2$-edge at $v$ and at $w$,
respectively. Hence,
the number of the left bounding $S_s^2$-edges increases by three. We claim
that there is no other increase for the number of left bounding $S_s^2$-edges.

Indeed, consider the part of $s'$ above $v$. Let $s''$ be the segment in $S_s^2$
that intersects $s'$ at $v$. Let $p'$ be the endpoint of $s'$ above $v$
and $p''$ be the endpoint of $s''$ above $v$. Note that both $p'$
and $p''$ are on $\partial\calP$. Consider the region $R(v)$ above $v$ enclosed by $\overline{vp'}$, $\overline{vp''}$, and the portion of $\partial\calP$ between $p'$ and
$p''$ (e.g., see Fig.~\ref{fig:arrangement}). Clearly, the region $R(v)$ is not in the zone of $s$. Further, $R(v)$ lies on the right of $\overline{vp'}$, and thus
$\overline{vp'}$ cannot contribute any left bounding $S_s^2$-edges to
$Z(s)$. In addition, if a left bounding $S_s^2$-edge $e$ that was in the zone
$Z(s)$ (before $s'$ is added) is intersected by $s'$ somewhere above $v$, then the part of $e$
to the right of $s'$ (i.e., the part of $e$ in the region $R(v)$) is
not in the zone $Z(s)$ any more after $s'$ is added. Hence, there is no increase in the number of
left bounding $S_s^2$-edges due to such an intersection.

In a similar way, we can show that the part of $s$ below $w$ does not
increase the number of left bounding $S_s^2$-edges in the zone $Z(s)$. Therefore,
after $s'$ is added, the total increase of the number of left bounding
$S_s^2$-edges is at most three.

We discuss above the case when $s'$ is the only segment in $S_s^2$ through the
rightmost intersection on $s$. If there is more than one such segment,
then we take an arbitrary such segment as $s'$. By a similar analysis as that above and that
in \cite{ref:deBergCo08}, we can show that the
total increase in the number of left bounding $S_s^2$-edges is at most five.
We omit the details.

We conclude that there are $O(m_2)$ $S_s^2$-edges and
$O(m_1+m_2)$ $S_s^1$-edges in the zone $Z(s)$.  The lemma thus follows.
\end{proof}

\section{Conclusions}
\label{sec:conclusions}

In this paper, we propose two new data structures for the weak
visibility query problem on a simple polygon, which improve upon the
previous work \cite{ref:AronovVi02,ref:BoseEf02,ref:BygiWe11}.
%Our techniques give many new geometric observations for bounding the
%time complexity of the query algorithms.
Some results (e.g., the ray-rotating data structure and the zone
theorem) may be of independent interest.

For the $O(k\log n)$ time queries, our first data structure is
clearly optimal. For the $O(k+\log n)$ time query, however,
an open question is whether a data structure of sub-cubic preprocessing time
and space is possible.

\end{document}